\documentclass[10pt, reqno, a4paper]{amsart}
\usepackage{amsthm, amsmath, amssymb}
\usepackage{chngcntr}
\usepackage{enumerate}
\usepackage{hyperref}

\newtheorem{theorem}{Theorem}
\newtheorem{corollary}[theorem]{Corollary}

\newcommand{\supp}{\operatorname{supp}} 
\newcommand{\cspace}{\ensuremath{\R^3 \times (-\pi R, \pi R)}}
\newcommand{\ef}{\ensuremath{e_\mathrm{4}^2}}
\newcommand{\et}{\ensuremath{e_\mathrm{3}^2}}
\newcommand{\parder}[1][r]{\ensuremath{\frac{\partial}{\partial #1}}}
\newcommand{\pardertwo}[1][r]{\ensuremath{\frac{\partial^2}{\partial #1^2}}}

\newcommand{\N}{\mathbb{N}}
\newcommand{\R}{{\mathbb{R}}}
\newcommand{\C}{{\mathbb{C}}}
\newcommand{\dd}{{{\rm d}}}
\newcommand{\ii}{{\rm i}}
\renewcommand{\H}{{\mathcal{H}}}
\newcommand{\Dom}{{\operatorname{Dom}}}
\newcommand{\cf}{\emph{cf.}}
\newcommand{\ie}{{\emph{i.e.}}}
\newcommand{\eg}{{\emph{e.g.}}}
\newcommand{\Zc}{1}
\newcommand{\Rc}{R_{\rm crit}}
\newcommand{\Hsa}{H_{\rm sa}}
\newcommand{\Hm}{H_{\rm min}}
\newcommand{\spes}{\sigma_{\rm ess}}
\newcommand{\Vc}{V_{\rm circ}}
\newcommand{\Hc}{H_{\rm circ}}

\begin{document}
\title[Hydrogen atom in space with a compactified extra dimension]{Hydrogen atom in space with a compactified extra dimension and potential defined by Gauss' law}

\author{Martin~Bure\v s}
\address[Martin~Bure\v s]{
Institute for Theoretical Physics and Astrophysics,
Masaryk University, Kotl\'a\v{r}sk\'a 2, 61137 Brno, Czech Republic
}
\email{bures@physics.muni.cz}
\thanks{We would like to thank Rikard von Unge, Klaus Bering, Jean-Marie Barbaroux and David Krej\v ci\v r\'ik for very useful discussions and suggestions. The first author was supported by the Czech government grant agency under contract no. GA\v CR 202/08/H072.}

\author{Petr Siegl}
\address[Petr Siegl]{Mathematical Institute, University of Bern, Sidlerstrasse 5, 3012 Bern, Switzerland
	\& On leave from Nuclear Physics Institute ASCR, 25068 \v Re\v z, Czech Republic}
\email{petr.siegl@math.unibe.ch}


\date{30th September 2014}

\maketitle

\begin{abstract}
We investigate the consequences of one extra spatial dimension for the stability and energy spectrum of the non-relativistic hydrogen atom with a potential defined by Gauss' law, \ie~proportional to~$1/|x|^2$.
The additional spatial dimension is considered to be either infinite or curled-up in a circle of radius $R$. In both cases, the energy spectrum is bounded from below for charges smaller than the same critical value and unbounded from below otherwise. 
As a consequence of compactification, negative energy eigenstates appear: if $R$ is smaller than a quarter of the Bohr radius, the corresponding Hamiltonian possesses an infinite number of bound states with minimal energy extending at least to the ground state of the hydrogen atom.
\end{abstract}



\section{Introduction} 

String theory, predicting the existence of extra dimensions, has stimulated the study of higher-dimensional models, which often provide a deeper insight into 3-dimensional physics. Surprisingly, in spite of many years of research, a question of what happens to ordinary atoms when more than the usual three dimensions are considered seems not fully answered. 

Many existing works deal with a general case of $d$-dimensional hydrogen atoms with the potential proportional to $1/|x|$, irrespective of the number of spatial dimensions. A system defined in this way is indeed stable and one can derive wave functions and their respective eigenenergies, see \eg~\cite{
alliluyev,
burgbacher,
nieto,
jaber97,
1995atom.ph..12001N,
Nepstad2006}.

However, as it has been pointed out \eg~in \cite{Morales-1996-57,Andrew_Supplee_1990,Braga:2005ai,gurevich,1990AmJPh..58.1183E,1985AmJPh..53..893A} or recently in \cite{2012arXiv1205.3740C}, a more physically relevant potential is the solution of Maxwell's equations for a point charge in the $d$-dimensional space. In details, $V_d(|x|)\sim |x|^{2-d}$ and the corresponding  Schr\"odinger equation reads
\begin{eqnarray}\label{schr.eq}
\left(-\frac{\hbar^2}{2m}\Delta- \frac{e^2_d}{|x|^{d-2}} \right) \psi= E \psi ,
\end{eqnarray}
where $e_d$ is a $d$-dimensional charge. Questions on stability and the influence of extra dimensions on energy spectra for this potential seem still open. 
Intuitively, an instability, or the unboundedness of the energy spectrum from below, could be viewed as a fall of the electron on the nucleus, \cf~\cite[p.~116]{landau} or \cite{A.1970}.
In this regard it has only been formally shown that there is no stable hydrogen atom in spaces where the additional dimensions are of infinite extent~\cite{Andrew_Supplee_1990,gurevich,Braga:2005ai}.
Nonetheless, it is natural to ask what happens if the extra dimensions are compactified. This work fills the gap by delivering rigorous results for the cases of both a compactified and a non-compactified extra dimension. 

We focus on spaces with one extra spatial dimension, which, among higher dimensional spaces, require a special treatment since the potential can be merged with the centrifugal term arising from the radial reduction of the central potential. In other words, because of the absence of a characteristic length, a procedure leading to dimensionless quantities, which works in the treatment of the radial equation for $d\neq 4$, cannot be used here \cite{Braga:2005ai,gurevich}.

First, we consider the extra dimension to be non-compactified, \ie~we treat it equally with the other three dimensions. Recalling standard tools of functional analysis, it is showed that for $Z \in [0,\Zc)$, where $Z:=2m\ef/\hbar^2$, 
the Hamiltonian is non-negative without any bound state solutions, \cf~Thm.~\ref{thm.stab.inf}. However, it does not make sense to speak about a stable atom since this system is always ionized. For $Z > \Zc$, the spectrum is unbounded from below, \cf~Thm. \ref{thm.uc.inst}. Hence, there is no stable hydrogen atom if the additional dimension is non-compactified. These results prove the formally derived conclusions in earlier works, see \eg~\cite{Braga:2005ai,gurevich,2012arXiv1205.3740C, Bures-2007}.

The main goal of this work is to investigate the consequences of circular \emph{compactification of the extra dimension} for the energy spectrum. More precisely, the underlying configuration space is $\R^3 \times S^1$, where $S^1$ is a circle of radius~$R$. For any potential with the singularity  $1/|x|^2$, \cf~\eqref{rozklad.V}, 
the stability result remains the same as in $\R^4$, \ie~the critical value $Z=1$, \cf~Thms. \ref{thm.stab.comp} and \ref{thm.c.inst}. Results on essential spectrum and existence of an infinite number of negative eigenvalues are presented in Thms.~\ref{thm.ess} and \ref{thm.bound.comp}.

The method of images, \cf~Section~\ref{subsec.moi}, provides a potential $\Vc$, \cf~\eqref{V.def}, with the desired behaviour:  $\Vc \sim 1/|x|^2$ close to the charge, \ie~it corresponds to the potential in $\R^4$, and the usual 3-dimensional behaviour $1/r$ is restored far from the charge. 
The resulting relation $\ef=2R\et$ between the three- and four-dimensional charges and the stability results, \cf~Section \ref{sec.comp}, imply the existence of a critical compactification radius (of the order of the Bohr radius $a_0$)
\begin{equation}\label{critical-radius}
\Rc:=\frac{a_0}{4}=\frac{\hbar^2}{4m\et}\approx 1.32\times 10^{-11} \mathrm{m},
\end{equation}
above which the system becomes instable. As a corollary of Thm.~\ref{thm.bound.comp}, for $R<\Rc$, the corresponding Hamiltonian possesses an infinite number of bound states. Furthermore, the minimal energy extends at least to the ground state energy of the hydrogen atom, \cf~Cor.~\ref{cor.ground.st}.

Because $\Rc$ was found to be of the order of the Bohr radius, we did not obtain a stronger constraint on models involving one extra dimension, as the latest calculations based on the LHC data, \cf~\cite{Datta-2014-89} and also~\cite{Belanger:2012mc,Kakuda:2013kba,Long-2003}, give the value of about~$10^{-18}$m as the upper bound. 
In more detail, the lower bound on the compactification scale $R^{-1}$ obtained at $3\sigma$ and coming from the LEP collider precision electroweak tests is about $R^{-1}> 260 \mathrm{GeV}$. The CMS data give a bound $R^{-1} > 720 \mathrm{GeV}$ at 95$\%$ confidence level. Recently, calculations based on the LHC data on the Higgs 
boson decay gave an even better constraint of $R^{-1} > 1.3 \mathrm{TeV}$ at 95$\%$ C.L. This corresponds to the radius of $9.5\times 10^{-19}$m which is much more restrictive than \eqref{critical-radius}.

Stronger theoretical constraints may be obtained by calculating corrections to the spectrum and comparing it to the hydrogen atom. An attempt of this type of calculations can be found in \cite{Floratos-2011-694}, however, only a different, simplified form of the potential than that corresponding to an extra compactified dimension, is treated eventually.
Perturbative calculations of the shifts of bound state energies due to the presence of an extra compactified dimension will be subject to a separate article~\cite{preparation}.

From the mathematical point of view, the essential tools for addressing the stability issue are the Hardy inequality, \cf~\cite{Hardy-1920-6}, \cite[Sec.~I.4.1]{Krejcirik-2010-LN} or \eqref{Hardy.ineq}, and the KLMN theorem \cf~\cite[Thm.~X.17]{Reed2}, \cite[Sec.~I.2.4]{Krejcirik-2010-LN} or Thm.~\ref{thm.KLMN} in Appendix; both are classical and frequently used in quantum problems, particularly for the stability of matter, see 
\eg~\cite{
lieb2010stability,
MR2499016}. 
The former can be viewed as a control of the potential energy by the kinetic energy. The latter (an analog of the Kato-Rellich theorem \cite[Thm.~X.12]{Reed2} for forms) allows to add small (in a certain sense) perturbations to a self-adjoint operator while preserving self-adjointness and semi-boundedness. For $Z \in [0,\Zc)$, the combination of these results yields the self-adjoint and non-negative Hamiltonian; as explained in Section \ref{stability.compact}, the compactified case can be treated analogously, \cf~Thm. \ref{thm.stab.comp}.
The definition of a self-adjoint Hamiltonian for $Z>\Zc$ is more complicated since KLMN theorem is no longer applicable (the perturbation is not small). Nonetheless, we construct suitable trial functions and show that any corresponding self-adjoint Hamiltonian has the spectrum unbounded below, \cf~Thms.~\ref{thm.uc.inst} and \ref{thm.c.inst}.
In the compactified and stable case, the existence of the infinite number of boundstates is obtained by a straightforward adaptation of the standard result \cite[Thm.~XIII.64]{Reed4} based on min-max principle, \cf~\cite[Sec.~XIII.1]{Reed4} or \cite[Sec.~I.2.4]{Krejcirik-2010-LN}.

The following notation is used:
$\Omega=\cspace$ is the compactified configuration space after the parametrization (of $S^1$), $B_{\rho}(x_0) := \{ x \in \R^4 : |x_0 - x| < \rho\}$  is the ball around $x_0 \in \R^4$ with radius $\rho$, $W^{k,p}(\Omega)$ are Sobolev spaces and $\|\cdot\|_{W^{k,p}}$ the corresponding norms, \eg~$\|\psi\|_{W^{1,2}}^2=\|\nabla\psi\|^2+\|\psi\|^2$, and the Euclidean norm in $\C^n$ is denoted by $|\cdot|$.

\section{Extra dimension of an infinite extent}\label{sec.inf}
First, we review the case of an extra dimension which is infinitely extended. The Hamiltonian for a closed system of two non-relativistic point masses interacting via a central force is $-\hbar^2\Delta/2m+\ef V(|x|)$, 
where $e_4$ is the four-dimensional charge (whose unit is $\text{energy}^{1/2} \times \text{length}$) and the potential $V$ is the general solution of Poisson's equation for a point charge in 4 dimensions, \ie~$V(x) = - 1/|x|^2$.

To simplify all formulas in the sequel, we define a dimensionless parameter $Z:=2m\ef/\hbar^2$, where $m$ is the reduced electron mass. 
Writing $k^2=2mE^2/\hbar^2$, the Schr\"odinger equation~\eqref{schr.eq} takes the form 
$(-\Delta-  Z /|x|^2) \psi= k^2 \psi$.

\subsection{Stability for weak charges: application of the Hardy inequality}\label{stability.uncompact}

We recall how to introduce the Hamiltonian in a mathematically correct way in the Hilbert space $L^2(\R^4)$. Although this approach might seem to be quite meticulous, the usage of the correct framework provides an immediate answer to the question of semi-boundedness of the spectrum of $H$. Moreover, it enables us to analyze the more complicated case with a compactified extra dimension.

We start from the free (self-adjoint and non-negative) Hamiltonian 
\begin{equation}
\begin{aligned}
H_0 :=-\Delta, \qquad 
\Dom(H_0):=W^{2,2}(\R^4).
\end{aligned}
\end{equation}
We consider $V$ as a perturbation of $H_0$ in a suitable sense, namely as a relatively form-bounded perturbation. In more detail, $H_0$ is associated (in the sense of the representation theorem \cite[Thm.~VI.2.1]{Kato-1966}, \cf~Thm.~\ref{thm.kato} in Appendix)
with the quadratic form 
\begin{equation}
h_0[\psi]:=\|\nabla \psi\|^2, \qquad \Dom(h_0):=W^{1,2}(\R^4),
\end{equation}
while $V(x)=-|x|^{-2}$ is associated with
\begin{equation}
	v[\psi]  := \int_{\R^4} V |\psi|^2, \quad \quad \Dom(v):=\left \{ \psi \in L^2(\R^4):  V |\psi|^2 \in L^1(\R^4) \right\}.
\end{equation}
The classical Hardy inequality \cite{Hardy-1920-6} reads (for $d\geq 3$)
\begin{equation}\label{Hardy.ineq}
\forall \psi \in W^{1,2}(\R^d),  \qquad \int_{\R^d} |\nabla \psi(x)|^2 \dd x \geq 
\frac{(d-2)^2}{4} \int_{\R^d} \frac{|\psi(x)|^2}{|x|^2} \dd x.
\end{equation}
Using the notation for quadratic forms, we get in our case, \ie~$d=4$, that
\begin{equation}\label{rel.bound.1}
\forall \psi \in \Dom(h_0), \qquad | v[\psi]| \leq h_0[\psi]. 
\end{equation}
This inequality means that the form $v$ is relatively bounded w.r.t. $h_0$ with the bound $1$ 
and it is the cornerstone of the definition of the Hamiltonian for $Z \in [0,\Zc)$. 
For $Zv$ is then relatively bounded w.r.t. $h_0$ with bound $Z<1$, thus the KLMN theorem is applicable, \cf~\cite[Thm.~X.17]{Reed2} or Thm.~\ref{thm.KLMN}. It yields that the quadratic form
\begin{equation}
h := h_0 + Z v, \qquad 
\Dom(h) := \Dom(h_0) = W^{1,2}(\R^4), \quad Z \in [0,1),
\end{equation}
is symmetric, closed, and bounded from below, hence associated with unique self-adjoint, bounded from below operator $H$ that represents our Hamiltonian.
In fact, having inequality \eqref{rel.bound.1}, we conclude that spectrum of $H$ is  non-negative.

Since the potential decays at infinity, using a standard argument based on construction of suitable singular sequences 
(regularized plane waves), it follows that $\sigma(H)=\spes(H)=[0,+\infty)$; for more details see \eg~the proof of Thm.~\ref{thm.ess}, dealing with an analogous problem. 
In other words, the spectrum of $H$ is the same as that of $H_0$,  
\ie~consisting of a branch of the continuous one without any negative eigenvalues.

Formally, we can sum up the facts above into the following statement:
\begin{theorem}\label{thm.stab.inf}
Let $V(x)=-|x|^{-2}$ and let $Z \in [0,\Zc)$. Then $H:= -\Delta +Z V$, defined as a sum of forms, is self-adjoint and non-negative.
Moreover, $\sigma(H)=\spes(H)=[0,+\infty)$.
\end{theorem}
In the border case $Z=1$, we can start with a symmetric Hamiltonian $\Hm$, \cf~\eqref{H.min}, that is non-negative due to Hardy inequality. Hence the Friedrichs extension, \cf~\cite[Thm.~X.23]{Reed2} or \cite[Thm.~4.6.11]{blank2008hilbert} for instance, is non-negative and thus represent a suitable Hamiltonian. 
The case $Z<0$ (repulsive interaction) is even simpler, the spectral result is the same as in Thm.\ref{thm.stab.inf}.

\subsection{Instability for strong charges} \label{instability.uncompact}
The instability of the Hamiltonian for $Z>\Zc$ can be explained with the help of facts related to the Hardy inequality. 
Before giving the details, we remark that there is a complication 
already in the definition of a self-adjoint Hamiltonian since, unlike in the $Z<\Zc$ case, the forms approach and KLMN theorem are not applicable. It is well-known, \cf~\cite[Sec.~X.2]{Reed2},
\cite{Krall-1982-45}, that we have an infinite number of self-adjoint operators that extend a natural minimal (only symmetric) operator
\begin{equation}\label{H.min}
\Hm
:=-\Delta + Z V, \qquad \Dom(\Hm) := C_0^{\infty}(\R^4 \setminus\{0\}).
\end{equation}
Therefore, in order to establish the instability, we have to show that any of these
possible Hamiltonians has spectrum unbounded from below. We provide a more abstract argument first and 
then we briefly recall the explanation due to Krall, \cf~\cite{Krall-1982-45}, suitable for extensions allowing separation of variables (in spherical coordinates).
\begin{theorem}\label{thm.uc.inst}
Let $Z> \Zc$ and let $\Hsa$ be any self-adjoint extension of operator $\Hm$ defined in \eqref{H.min}.
Then the spectrum of $\Hsa$ is unbounded from below (and from above).
\end{theorem}
\begin{proof}
We use the fact that there exists an optimizing sequence of functions $\{\psi_n\} \subset W^{1,2}(\R^4)$ for the Hardy inequality 
\eqref{Hardy.ineq}. 
Let $0< \delta < 1/2$, and let $\eta \in \C_0^{\infty}(\R)$ be such that $0\leq \eta \leq  1$, $\eta\equiv 1$ for $(-\delta,\delta)$ and $\supp \eta \subset (-3 \delta/2, 3 \delta/2)$. We define a sequence of functions  
\begin{equation}\label{trial.compact}
\psi_n(x):=|x|^{-1 + 1/n} \eta(|x|), 
\end{equation}
belonging to $W^{1,2}(\R^4)$. Since $C_0^{\infty}(\R^4 \setminus \{0\})$ is dense in $W^{1,2}(\R^4)$ (in $W^{1,2}$ norm), \cf~\cite[Cor.~VIII.6.4]{EE}, we can construct, using a standard diagonal scheme
argument, a sequence of functions $\{\varphi_n\} \subset C_0^{\infty}(B_{2\delta}(0) \setminus \{0\})$ such that 
\begin{equation}\label{ineqn4}
\begin{aligned}
\| \psi_n-\varphi_n\|_{W^{1,2}}^2<\frac1{n}.
\end{aligned}
\end{equation}
Moreover, from the Hardy inequality~\eqref{Hardy.ineq}, we get:
\begin{equation}\label{ineqn5}
\begin{aligned}
\int_{\R^4} \frac{|\psi_n-\varphi_n|^2}{|x|^2} \leq
\int_{\R^4} \left| \nabla (\psi_n-\varphi_n) \right|^2 \dd x \leq
\| \psi_n-\varphi_n\|_{W^{1,2}}^2<\frac1{n}.
\end{aligned}
\end{equation}

The concluding part of the proof consists of a couple of inequalities. All constants $C_i$ are positive real numbers. Firstly, it is easy to verify that the normalization of $\psi_n$ is bounded as follows:
\begin{equation}\label{ineqn1}
0<C_1\leq \|\psi_n\|^2 \leq C_2 < +\infty.
\end{equation}
Further, for all $n\in\mathrm{N}$ we have
\begin{equation}\label{ineqn2}
\begin{aligned}
&\left|\int_{\R^4} |\nabla \psi_n|^2 - \int_{\R^4} \frac{|\psi_n|^2}{|x|^2}\right|
=C_\Theta\left|\int_0^\infty \left|\partial_\rho \left( \rho^{-1+1/n}\eta(\rho)\right)\right|^2 \rho^3 \dd \rho 
-\int_0^\infty  \left|\rho^{-1+1/n}\eta(\rho)\right|^2 \rho \dd \rho \right| 
\\
&=C_\Theta\left| \frac{1-2n}{n^2} 
\int_0^\infty \rho^{-1+2/n} \eta^2(\rho)\dd \rho 
+
\int_0^\infty \rho^{1+2/n} \eta'(\rho)^2 \dd \rho
+
\frac{2(1-n)}{n}\int_0^\infty \rho^{2/n}\eta(\rho)\eta'(\rho)\dd \rho  \right| 
\\
& \leq C_\Theta\left(
\frac{2}{n} \int_0^{2\delta} \rho^{-1+2/n} \dd \rho
+ \tilde{C}_1 \int_0^{2\delta} \rho^1 \dd \rho
+ \tilde{C}_2 \int_0^{2\delta} \rho^0 \dd \rho
\right)
\leq C_3 < +\infty,
\end{aligned}
\end{equation}
where $C_\Theta$ comes from integrating over angular coordinates. 

Finally, taking, as a lower bound, the integral over $(0,\delta)$ in the radial part, we easily obtain
\begin{equation}\label{ineqn3}
\int \frac{|\psi_n|^2}{|x|^2} \geq n C_4.
\end{equation}

We are now ready to analyze the infimum of the spectrum of $\Hsa$.
Since $\varphi_n \in \Dom(\Hm) \subset \Dom(\Hsa)$, the integration by parts yields
\begin{equation}
\frac{\langle \varphi_n, \Hsa \varphi_n \rangle}{\|\varphi_n\|^2} = 
\frac{\| \nabla\varphi_n\|^2  + Z \langle \varphi_n, V \varphi_n \rangle}{\|\varphi_n\|^2}.
\end{equation}
In the next steps, we approximate $\varphi_n$ by $\psi_n$. At first, using $|a-b|\geq|a|-|b|$, Young inequality (with setting $\varepsilon = 1/2$  eventually), \eqref{ineqn4} and  \eqref{ineqn1}, we obtain
\begin{equation}\label{ineqn6}
\begin{aligned}
\|\varphi_n\|^2 &= \|\varphi_n-\psi_n+\psi_n\|^2 
\geq
(1-\varepsilon)\|\psi_n\|^2+\left(1-\frac 1 \varepsilon \right) \| \varphi_n -\psi_n\|^2 
\geq
\frac1{2}\|\psi_n\|^2- \| \varphi_n -\psi_n\|_{W^{1,2}}^2 
\\ & \geq \frac{C_1}2 - \frac 1 n
.
\end{aligned}
\end{equation}
Next, using again triangle, Young inequalities,  $|a-b|\geq|a|-|b|$, \eqref{ineqn4}, \eqref{ineqn5} and \eqref{ineqn2}, we get
\begin{equation}\label{ineqn7}
\begin{aligned}
\langle \varphi_n, \Hsa \varphi_n \rangle
& =
\|\nabla\varphi_n -\nabla\psi_n+\nabla\psi_n\|^2 - Z \int_{\R^4}\frac{|\varphi_n-\psi_n+\psi_n|^2}{|x|^2}
\\
&\leq \frac{n+1}{n} \|\nabla \psi_n\|^2 + (n+1)\|\nabla\varphi_n -\nabla\psi_n\|^2 
- Z \frac{n-1}{n}   \int_{\R^4} \frac{ |\psi_n|^2 - n|\varphi_n-\psi_n|^2 } {|x|^2} 
\\
&\leq \frac{n+1}{n}  \left( \|\nabla \psi_n\|^2 - \int_{\R^4} \frac{|\psi_n|^2}{|x|^2}\right) 
- \left( Z-1 - \frac{Z +1}{n} \right) \int_{\R^4} \frac{|\psi_n|^2}{|x|^2}
\\
& \quad + (n+1)\| \psi_n-\varphi_n\|^2_{W^{1,2}} + Z(n-1)\int_{\R^4}\frac{|\varphi_n-\psi_n|^2}{|x|^2} 
\\
&\leq \frac{n+1}{n} C_3 - n\, C_4 \left( Z-1 - \frac{Z +1}{n}\right) + \frac{n+1}{n} +Z\frac{n-1}{n}.
\end{aligned}
\end{equation}
In the last step, we combine \eqref{ineqn6} and \eqref{ineqn7} and we obtain 
\begin{equation}
\frac{\langle \varphi_n, \Hsa \varphi_n \rangle}{\|\varphi_n\|^2} \to -\infty \quad \mbox{as } \quad n \to \infty.
\end{equation}
Hence, the spectrum of any self-adjoint extension $\Hsa$ of $\Hm$ is unbounded from below.
\end{proof}

\subsubsection{Spectrum of spherically symmetric extensions}
The symmetric operator $\Hm$ can be rewritten in spherical coordinates $(\rho,\phi_1,\phi_2,\phi_3)$
as

\begin{equation}
\Hm = -\pardertwo[\rho]-\frac3{\rho}\parder[\rho]+\frac{\mathcal{L}^2(3)}{\rho^2} - \frac{Z}{\rho^2}  ,
\end{equation}
where $\mathcal{L}^2(3)$
is the square of the angular momentum operator on a 3-sphere, see \eg~\cite{Bures-2007} for details. 
Due to spherical symmetry of the potential, we can separate variables by making the Ansatz $\psi(x)=R(\rho)Y^{(3)}(\phi_1,\phi_2,\phi_3)$, 
and thus obtain the radial operator acting in $L^2((0,\infty),\rho^3 \dd \rho)$:
\begin{equation}
-\frac{\dd^2 }{\dd \rho^2} - \frac{3}{\rho} \frac{\dd }{\dd \rho} + \frac{l(l+2)-Z}{\rho^2}.
\end{equation}
The usual transformation $R(\rho) \mapsto \rho^{-3/2} \tilde{R}(\rho)$ yields the final radial operator 
in $L^2((0,\infty),\dd \rho)$:
\begin{equation}
\Hm^{\rm rad}:= -\frac{\dd^2 }{\dd \rho^2} -  \frac{\gamma}{\rho^2},
\end{equation}
where $\gamma = Z-3/4 - l(l+2)$, which is exactly the form of Hamiltonian studied in \cite{PhysRev.80.797} and particularly in \cite{Krall-1982-45} for $\gamma > 1/4$, \ie~$Z>1$ if $l=0$.  
A one parameter family of self-adjoint operators $H_\alpha$ (possible Hamiltonians) having the same action as $\Hm^{\rm rad}$ when restricted to $C_0^{\infty}(\R_+)$ functions is found there; the parameter $\alpha$ enters the generalized boundary conditions at $0$ that describe the domain of $H_{\alpha}$. The spectrum of any $H_{\alpha}$ contains continuous branch $[0,+\infty)$ and negative eigenvalues having accumulation points at $0$ and at $-\infty$, moreover, the algebraic eigenvalue equation, where $\alpha$ enters as a parameter, is known, \cf~\cite{Krall-1982-45}.

\section{Compactified extra dimension} \label{sec.general.ed}
We consider a configuration space with the ordinary three extended spatial dimensions $(x_1,x_2,x_3)\in\R^3$ and one extra compactified dimension $x_4$ with the topology of a circle with radius~$R$, \ie~$\R^3 \times S^1$.
Parameterizing $S^1$, we obtain the configuration space $\Omega = \cspace$ and the Hamiltonian 
$H = -\Delta + Z V$ in $L^2(\Omega)$ with periodic boundary conditions at $(x_1,x_2,x_3,\pm \pi R)$ in the domain.
We address the question of stability for Hamiltonians with the potential:
\begin{eqnarray}\label{rozklad.V}
V(x) =- \frac1{|x|^2}+ W(x), \quad \mbox{with } W \in L^{\infty}(\Omega).
\end{eqnarray}
The motivation for this choice of configuration space will become clear in Section~\ref{sec.comp} where the particular form of $W$ will be derived.

\subsection{Stability and instability}
\label{stability.compact}

As in Section~\ref{stability.uncompact}, we will treat the problem by introducing quadratic forms corresponding to our operators.
The kinetic part $H_0:=-\Delta$ is associated with the form 
\begin{equation}\label{h0.def.comp}
h_0[\psi]:=\|\nabla \psi\|^2, \qquad  \Dom(h_0):=\{\psi \in W^{1,2}(\Omega): \psi(x_1,x_2,x_3,-\pi R)=\psi(x_1,x_2,x_3,\pi R)\}
\end{equation}
while $V$ having the form \eqref{rozklad.V} is associated with
\begin{equation}
v[\psi]:=\int_{\Omega} V |\psi|^2, \quad \quad \Dom(v):=\left \{ \psi \in L^2(\Omega): V |\psi|^2 \in L^1(\Omega) \right \}.
\end{equation}
\begin{theorem}\label{thm.stab.comp}
Let $V$ be given by~\eqref{rozklad.V}. Then, for all $\psi \in W^{1,2}(\Omega)$ and all $\varepsilon >0$,
\begin{equation}\label{c.I}
\int_\Omega |V||\psi|^2  \leq (1+\varepsilon ) \int_\Omega |\nabla \psi|^2 + C(\varepsilon) \|\psi\|^2,
\end{equation}
where $C(\varepsilon) >0$ is independent of $\psi$.
Hence, if $Z \in [0,\Zc)$, then $H:= -\Delta +Z V$, defined as a sum of forms, is self-adjoint and bounded from below.
\end{theorem}
\begin{proof}
Let $\xi \in C_0^{\infty}(\Omega)$, where $0 \leq \xi \leq 1$ and $\xi \equiv 1$ for $x\in B_{R-\delta}(0)$, where $\delta>0$, \ie~in a ball with radius $R-\delta$ around the singularity.
Then 
\begin{equation}
\begin{aligned}
\int_\Omega |V||\psi|^2 
& \leq    \int_\Omega \frac{|\xi \psi|^2}{|x|^2} + \int_\Omega \frac{(1-\xi^2)|\psi|^2}{|x|^2}  + \int_\Omega |W| |\psi|^2
\leq \int_\Omega \frac{|\xi \psi|^2}{|x|^2} + C_1 \|\psi\|^2
\leq  \int_\Omega |\nabla(\xi \psi)|^2 + C_1 \|\psi\|^2 
\\& = \int_\Omega |(\nabla\xi) \psi + \xi (\nabla \psi)|^2 + C_1 \|\psi\|^2  
\leq  (1+\varepsilon ) \int_\Omega |\xi|^2|\nabla \psi|^2 
+ \left(1+\frac1{\varepsilon}\right) \int_\Omega |\nabla \xi|^2|\psi|^2 
+ C_1\|\psi\|^2
\\
&\leq (1+\varepsilon ) \int_\Omega |\nabla \psi|^2 + C(\varepsilon) \|\psi\|^2,
\end{aligned}
\end{equation}
with arbitrary positive~$\varepsilon$. 
The last inequality on the first line is due to the Hardy inequality~\eqref{Hardy.ineq} and the fact that $\xi\psi \in W^{1,2}(\R^4)$.
The inequality on the second line is due to the triangle and the Young inequality.
Now, multiplying both sides of~\eqref{c.I} by $Z$, we get using the notation for quadratic forms:
\begin{equation}
\forall \psi \in \Dom(h_0), \quad \quad |v[\psi]| \leq  Z(1+\varepsilon ) h_0[\psi]  + ZC(\varepsilon) \|\psi\|^2.
\end{equation}
If $Z \in [0,\Zc)$, then $Zv$ 
is relatively bounded w.r.t. $h_0$ with bound $Z \in [0,\Zc)$, because $\varepsilon$ is an arbitrary positive number.
Hence, we can apply the KLMN theorem yielding the claim.
\end{proof}
Let us note that the conclusion on the stability is valid also for a finite number of singularities of the type $1/|x|^2$, \ie~$\sum_{k=0}^{N}=1/|x-c_k|^2$, where $c_k\in \Omega$. The function $W$ may not be bounded, for instance, a relative form-boudedness with respect to~$h_0$ with bound 0 is sufficient. Also, the stability result remains unchanged if we replace the periodic boundary conditions by Dirichlet, Neumann or Robin ones. Moreover, the proof can be extended to any sufficiently regular domain $\Omega$ in a standard way.

For $Z>\Zc$, the system is not stable, as in the uncompactified case, \cf~Thm.~\ref{thm.uc.inst}. Again, we start with a symmetric minimal operator
\begin{equation}\label{Hmin.comp}
\Hm :=-\Delta + Z V, \qquad \Dom(\Hm) := \{ \psi \in C_0^{\infty}(\Omega \setminus \{0\})\},
\end{equation}
and deal with its self-adjoint extensions.
\begin{theorem}\label{thm.c.inst}
Let $Z>\Zc$ and let $\Hsa$ be any self-adjoint extension of a symmetric operator $\Hm$ defined in \eqref{Hmin.comp}.
Then the spectrum of $\Hsa$ is unbounded from below (and from above).
\end{theorem}
\begin{proof}
The sequence $\varphi_n$ of trial functions from the proof of Thm.~\ref{thm.uc.inst} was intentionally constructed in such a way that $\supp \varphi_n \subset (B_{2 \delta} \ \{0\} )$ with arbitrary $\delta>0$, hence we can use it also in $\Omega$.
So it follows that, for $Z>1$, the spectrum of $\Hsa$ is unbounded from below.
\end{proof}

\subsection{Essential spectrum}

For $Z \in [0,\Zc)$ and $V$ decaying at infinity, the essential spectrum of $H$ remains $[0,+\infty)$, as in the case of an extra dimension of infinite extent, \cf~Thm.~\ref{thm.stab.inf}. 

\begin{theorem}\label{thm.ess}
Let $Z \in [0,\Zc)$, $V$ be given by~\eqref{rozklad.V} with $W(x) \to 0$ as $|x| \to + \infty$ and let $H = -\Delta +Z V$ be defined as a sum of forms. Then
$\spes(H)=[0,+\infty)$.
\end{theorem}
\begin{proof} 
Since the potential $V$ decays at infinity, the interval $[0,+\infty)$ belongs to $\spes(H)$ as it can be seen from Weyl criterion, \cf~\cite[Thm.~VII.12]{Reed1} or Thm.~\ref{thm.weyl} in Appendix, and a usual construction of singular sequences, \cf~for instance~\cite[Sec.~I.3.2, II.2.2]{Krejcirik-2010-LN}, \cite{Krejcirik-2005-41} or \cite[Sec.~8.3]{Davies-1995}. In more detail, define functions
\begin{equation}\label{fcts.weyl}
\psi_n(x_1,x_2,x_3,x_4) := \varphi_n (x_1,x_2,x_3) \, e^{\ii (k,k,k)_.(x_1,x_2,x_3)},
\end{equation}
where 
\begin{equation}
\varphi_n(x_1,x_2,x_3):= n^{-3/2} \varphi 
\left(
\frac{x_1}{n}-n,\frac{x_2}{n}-n,\frac{x_3}{n}-n
\right)
\end{equation}
with some $\varphi \in C_0^{\infty}(\R_+^3)$ and $\|\varphi\|=1$. Functions $\varphi_n$ and $\psi_n$ satisfy $\supp \varphi_n \subset \overline{\Omega}_n$, where $\Omega_n  := (n,+\infty)^3 \times (-\pi R, \pi R)$, 
and 
\begin{equation}\label{fcts.weyl.prop}
\| \nabla \varphi_n\| = \frac{\| \nabla \varphi\|}{n}, 
\qquad 
\| \Delta \varphi_n\| = \frac{\| \Delta \varphi\|}{n^2}, 
\qquad
\|\varphi_n\| = 1,
\qquad
\|\psi_n\| = 1.
\end{equation}
Moreover, using the first representation theorem \cite[Thm.~VI.2.1]{Kato-1966} part $iii)$, \cf~Thm.~\ref{thm.kato} in Appendix,  we show that $\psi_n \in \Dom(H)$. Clearly, $\psi_n \in \Dom(h) = \Dom(h_0)$, \cf~\eqref{h0.def.comp}. Taking an arbitrary $\varphi\in \Dom(h)$ and integrating by parts, we obtain 
$h(\varphi, \psi_n)=
\langle\nabla \varphi, \nabla \psi_n \rangle + Z v(\varphi,\psi_n)
= \langle \varphi, -\Delta \psi_n + Z V \psi_n \rangle.$
Since $-\Delta \psi_n - Z V \psi_n \in L^2(\Omega)$, we get (by the first representation theorem) that $\psi_n \in \Dom(H)$ .

Finally, a direct computation gives
\begin{equation}
|-\Delta \psi_n-k^2\psi_n| = |\Delta\varphi_n|^2 +4|k \cdot \nabla\varphi_n|^2 
\leq
|\Delta\varphi_n|^2 +4k^2 |\nabla\varphi_n|^2,
\end{equation}
which implies, using~\eqref{fcts.weyl.prop},
\begin{equation}
\| -\Delta \psi_n-k^2\psi_n\| 
\leq
\|\Delta\varphi_n\|^2 +4k^2\|\nabla\varphi_n\|^2 \xrightarrow{\scriptscriptstyle n\to\infty} 0.
\end{equation}
Hence 
\begin{equation}
\begin{aligned}
\|-\Delta \psi_n + Z V \psi_n - k^2 \psi_n\| 
& \leq 
\|-\Delta \psi_n - k^2 \psi_n\| + \|Z V \psi_n\| 
 \leq 
\|-\Delta \psi_n - k^2 \psi_n\| + |Z| \|V\|_{L^{\infty}(\Omega_n)} 
\\
&\leq
\sqrt{ \|-\Delta \varphi_n\|^2  + 4 k^2 \|\nabla \varphi_n\|^2 }
+ 
|Z| \|V\|_{L^{\infty}(\Omega_n)} \xrightarrow{\scriptscriptstyle n\to\infty} 0,
\end{aligned}
\end{equation}
so $[0,\infty) \subset \sigma(H)$ by Weyl's criterion, \cf~Thm.~\ref{thm.weyl}. Since the interval $[0,\infty)$ has clearly no isolated points, it must be in $\spes(H)$.

On the other hand, $\spes(H) \subset [0,+\infty)$ since $\inf \spes(H) \geq 0$. The latter can be justified by the so-called min-max principle, \cf~\cite[Sec.~XIII.1]{Reed4} or~\cite[Sec.~I.2.4]{Krejcirik-2010-LN}, and Neumann bracketing, \cf~\cite[Sec.~XIII.15]{Reed4} or~\cite[Sec.~II.2.1, II.2.2]{Krejcirik-2010-LN}. The latter consists of dividing $\Omega = \Omega_n \cup \tilde \Omega_n \cup \Sigma_n$,
where $\tilde \Omega_n := (-n,n)^3 \times (-\pi R, \pi R) $ and $\Sigma_n$ represents a new boundary, and considering two operators $H_{\rm int}^N$, $H_{\rm ext}^N$ acting in $L^2(\tilde \Omega_n)$, $L^2(\Omega_n)$, respectively, being ``restrictions'' of $H$ with additional Neumann boundary conditions on $\Sigma_n$. More precisely, these operators are defined via quadratic forms
\begin{equation}
\begin{aligned}
h_{\rm int}^N[\psi] & := \|\nabla \psi\|^2 + Z v[\psi], 
\quad 
\Dom(h_{\rm int}^N) := \{ \psi \in W^{1,2}(\tilde \Omega_n)  : \psi(x_1,x_2,x_3,-\pi R)=\psi(x_1,x_2,x_3,\pi R) \},
\\
h_{\rm ext}^N[\psi] & := \|\nabla \psi\|^2 + Z v[\psi], 
\quad 
\Dom(h_{\rm ext}^N) := \{ \psi \in W^{1,2}(\Omega_n)  : \psi(x_1,x_2,x_3,-\pi R)=\psi(x_1,x_2,x_3,\pi R) \}.
\end{aligned}
\end{equation}
Since $W^{1,2}(\tilde \Omega_n)$ is compactly embedded in $L^2(\tilde \Omega_n)$, \cf~\cite[Thm.~6.3]{Adams1975}, the resolvent of $H_{\rm int}^N$ is compact, \cf~\cite[Thm.~XIII.64]{Reed4}, therefore $\spes(H_{\rm int}^N) = \emptyset$. Finally, the min-max principle yields
\begin{equation}
\inf \spes(H) \geq \inf \spes(H_{\rm ext}^N) \geq \inf \sigma(H_{\rm ext}^N) \geq  -|Z|\| V\|_{L^{\infty}(\Omega_n)} \to 0,
\end{equation}
where the last inequality follows from the estimate 
$h_{\rm ext}^N[\psi]\geq Z\int_{\Omega_n} V|\psi|^2 \geq -|Z|\| V\|_{L^{\infty}(\Omega_n)}  \|\psi\|^2$.
\end{proof}

\subsection{Negative eigenvalues}

Unlike the case of the infinitely extended extra dimension, \cf~Thm.~\ref{thm.stab.inf}, $H$ may posses an infinite number of bound states.

\begin{theorem}\label{thm.bound.comp}
Let $Z \in (0,\Zc)$, $V$ be given by~\eqref{rozklad.V} and satisfying
\begin{equation}
V(x) \leq - \frac{a}{(x_1^2+x_2^2+x_3^2)^{1 - \varepsilon}}, \quad \mbox{ if } |x| \geq \rho_0,
\end{equation}
for some $a$, $\varepsilon$, $\rho_0 >0$.
Then the negative spectrum of $H = -\Delta + Z V$, defined as a sum of forms, consists of an infinite number of eigenvalues. 
\end{theorem}

\begin{proof}
The claim follows from a simple and straightforward adaptation of the proof of \cite[Thm.~XIII.6]{Reed4}.  Take a non-negative function $\phi \in C_0^{\infty}(\R^3)$ with $\supp \phi \subset \{ x \in \R^3 \, : \: 1 < r < 2\}$, where $r^2:=x_1^2 + x_2^2 + x_3^2$, and $\|\phi\|_{L^2(\R^3)}=1$. For $\rho >0$, define $\psi_{\rho}(x_1,x_2,x_3,x_4):= (2 \pi R)^{-1/2} \rho^{-3/2} \phi(x_1 \rho^{-1},x_2 \rho^{-1},x_3 \rho^{-1})$. Hence, $\|\psi_\rho\|=1$, $\supp \psi_\rho \subset \{ x \in \Omega \, : \, \rho < r < 2 \rho \}$. Clearly $\psi_\rho \in \Dom(h)$ and, similarly as in the proof of Thm.~\ref{thm.ess}, it can be verified that $\psi_\rho \in \Dom(H)$. Moreover, for $\rho > \rho_0$,
\begin{equation}
h[\psi_\rho] = 
\langle \psi_\rho, -\Delta \psi_\rho \rangle + Z  \langle \psi_\rho, V \psi_\rho \rangle
\leq \rho^{-2} \langle \psi_1, - \Delta \psi_1 \rangle - Z a \rho^{-2 + 2 \varepsilon } \langle \psi_1, r^{-2 + 2\varepsilon} \psi_1 \rangle,
\end{equation}
where the last expression is negative for large $\rho$. Taking a sequence $\{\rho_n\}_n$ such that $\psi_{\rho_n}$ have disjoint supports, we obtain an orthonormal sequence of trial functions for which $\langle \psi_{\rho_n}, H \psi_{\rho_n} \rangle < 0$ and $\langle \psi_{\rho_n}, H \psi_{\rho_m} \rangle =0$ if $n\neq m$. The proof is concluded by the Rayleigh-Ritz principle, \cf~\cite[Thm.~XIII.3]{Reed4} and the proof of \cite[Thm.~XIII.6]{Reed4} for further details.
\end{proof}

\section{Hamiltonian of circular compactification} \label{sec.comp}
We derive a particular potential $\Vc$ corresponding to two charged particles in a space with circularly compactified extra dimension and apply previous stability and spectral results; the latter yields a critical compactification radius $\Rc$. Finally, we obtain an upper bound for the ground state energy of the corresponding Hamiltonian $\Hc$.

\subsection{Method of images and stability}
\label{subsec.moi}

The potential $\Vc$ is obtained from the so-called method of images, \cf~\cite{Bures-2007} for further details. The method consists in expanding, or "unrolling", the compactified dimension to get an infinite periodic space which repeats itself with the period of~$2\pi R$, where $R$ is the radius of the compactified dimension. In this way we get an infinite "chain" of systems periodically spaced along the extra dimension. 
If we again denote $(x_1,x_2,x_3):=\vec{r}$ coordinates in the extended 3-dimensional space and $x_4$ the coordinate along the compactified dimension (now unrolled), 
we have, instead of the potential of a single charge, the potential of a "chain" of charges placed at points
$(\vec{0},0)$, $(\vec{0},2\pi R)$, $(\vec{0},4\pi R)$, etc.
Hence, we calculate the overall potential simply by summing up the four-dimensional potentials
$-1/|x|^2$, \cf~\eqref{schr.eq}:
\begin{equation}\label{V.def}
	\Vc (x) := -\sum_{n=-\infty}^{\infty} \frac1{x_1^2 + x_2^2 + x_3^2 + (x_4-2 \pi n R)^2  } 
    =  - \frac1{2R r} \frac{\sinh(r/R)}{\cosh(r/R)- \cos(x_4/R)},
\end{equation}
where the equality can be showed with the help of the residue theorem, \cf~\cite[appendix C]{Bures-2007}.

For $r\ll R$ and $x_4\ll R$,  $\Vc \sim -1/(r^2+x_4^2)$, so, close to the charge, $\Vc$ behaves as in the uncompactified case. On the other hand, for $r$ large, $\Vc \sim -1/(2rR)$, \ie~the usual three-dimensional behaviour is restored. 
By comparing potential energies, \ie~$-e_3^2/r$ and $-\ef/(2rR)$, in the latter region,
we establish the relation between both charges:
\begin{equation}\label{rel.charges}
\ef=2R\et.
\end{equation} 
Hence, the parameter $Z$ is expressed by means of the Bohr radius $a_0=\hbar^2/m\et$: 
\begin{equation}\label{Z.R.rel}
Z=\frac{2m\ef}{\hbar^2}=\frac{4Rm\et}{\hbar^2}=\frac{4R}{a_0}.
\end{equation}

Combining the latter and Thms.~\ref{thm.stab.comp}, \ref{thm.c.inst}, we determine the critial radius $\Rc = a_0 /4$ above which the system becomes instable:
\begin{corollary}\label{cor.compact}
	The potential $\Vc$, defined in \eqref{V.def}, satisfies \eqref{rozklad.V} with $W(x) \to 0$ as $|x| \to +\infty$.
	Hence, for $R \in (0,\Rc)$, the corresponding Hamiltonian $$\Hc  = -  	\Delta + \frac{4R}{a_0} \Vc,$$ 
	defined as a sum of forms, is bounded from below and $\spes(\Hc)=[0,+\infty)$. 
	For $R > \Rc$, any self-adjoint extension of a symmetric minimal operator $\Hm$, defined in \eqref{Hmin.comp} with $V = \Vc$, has spectrum unbounded from below (and above).
\end{corollary}

\begin{proof}
We need to verify the assumptions of Thms.~\ref{thm.stab.comp}, \ref{thm.c.inst} and~\ref{thm.ess}, where~\eqref{Z.R.rel} is used.
Clearly, $\Vc(x) = - |x|^{-2} + W(x)$, where
\begin{equation}\label{fct.W}
\begin{aligned}
W(x):=-\underset{n\neq 0}{\sum_{n=-\infty}^{\infty}}
\frac1{r^2+ (x_4-2 \pi n R)^2  }.
\end{aligned}
\end{equation}
Simple estimates show that $W$ is bounded:
\begin{equation}\label{inequality-bound}
\begin{aligned}
|W(x)| & =
\underset{n\neq 0}{\sum_{n=-\infty}^{\infty}}\frac1{r^2+(x_4-2\pi R n)^2}\leq
\underset{n\neq 0}{\sum_{n=-\infty}^{\infty}}\frac1{(x_4-2\pi R n)^2}\leq
\underset{n\neq 0}{\sum_{n=-\infty}^{\infty}}\frac1{(\pi R-2\pi R|n|)^2}
\\
&=\frac1{\pi^2 R^2} \underset{n\neq 0}{\sum_{n=-\infty}^{\infty}}\frac1{(2|n|-1)^2}
=\frac2{\pi^2 R^2} \sum_{n=1}^{\infty}\frac1{(2n-1)^2}
=\frac2{\pi^2 R^2} \sum_{n=0}^{\infty}\frac1{(2n+1)^2}
=\frac1{4R^2},
\end{aligned}
\end{equation}
because $\sum_{n=0}^{\infty}\frac1{(2n+1)^2}=\lambda(2)=\pi^2/8$, see \eg~\cite[p.~368]{arfken2005mathematical}. 
Moreover, $W(x) \to 0$ as $|x| \to + \infty$.
\end{proof}

\subsection{Bound states for $R<\Rc$}\label{spec.ham.cir}
The potential $\Vc$ satisfies the assumptions of Thm.~\ref{thm.bound.comp}, \cf~the right hand side  of \eqref{V.def}, hence $\Hc$ has an infinite number of negative eigenvalues. Moreover,  $\Vc$ has a special property

\begin{equation}\label{magic-formula}
\int_{-\pi R}^{\pi R} \Vc(x)\dd x_4=-\frac{\pi}{r},
\end{equation}
which can be seen from (\cf~the integral formula~\cite[2.553,3]{gradshteyn2007})
\begin{equation*}
\frac{\sinh(r/R)}{2R} \int_{-\pi R}^{\pi R} \frac{\dd x_4}{\cosh(r/R)- \cos(x_4/R)}
=\arctan \left(\frac{(\cosh(r/R)-1)\tan(x_4/2R)}{\sinh(r/R)}\right)\Biggr|_{-\pi R}^{\pi R}=\pi.
\end{equation*}
The property \eqref{magic-formula} indicates that the eigenfunctions of the usual three dimensional hydrogen, 
\begin{equation}\label{3d-hyd}
\phi_{Nlm}(r,\theta,\varphi)=C_{Nlm}r^l e^{-\frac{r}{Na_0}} L_{N-l-1}^{2l+1}\left(\frac{2r}{Na_0}\right) P_l^m(\cos\theta) e^{\ii m \varphi},
\end{equation}
where $N\in\N$, $l\in\{0,1,\dots N-1\}$, $m\in\{-l\dots l\}$ and $C_{Nlm}$ are normalization constants, become suitable trial functions for negative spectrum of $\Hc$. For instance, they can be used to obtain an alternative proof of  the existence of the infinite number of negative eigenvalues of $\Hc$. For our purposes, we use $\phi_{100}$ to find an upper bound for the ground state energy of $\Hc$.
\begin{corollary}\label{cor.ground.st}
Let $R \in (0,\Rc)$ and let $\Hc$ be as in Cor.~\ref{cor.compact}.
Then the negative spectrum of $\Hc$ consists of an infinite number of eigenvalues and $\min \sigma(\frac{\hbar^2}{2m} \Hc) \leq-me_3^4/(2\hbar^2)$.
\end{corollary}

\begin{proof}
The existence of the infinite number of negative eigenvalues follows from \eqref{V.def} and Thm.~\ref{thm.bound.comp}. 
To show the upper bound of $\min \sigma(\frac{\hbar^2}{2m} \Hc)$, define $\psi:=(2 \pi R)^{-1/2} \phi_{100}$, \cf~\eqref{3d-hyd}. It is not difficult to verify that $\psi \in \Dom(h)$ and $\|\psi\| =1$. Moreover,
\begin{equation} \label{neg}
h[\psi]=
h_0[\psi]-\frac{4R}{a_0} v[\psi]=
\|\nabla_3 \phi_{100}\|^2_{L^2(\R^3)}-\frac{2}{a_0}\int_{\R^3}\frac{|\phi_{100}|^2}{r}=
-\frac1{a_0^2},
\end{equation}
where the last step follows from the fact that $\phi_{100}$ is the ground state of 3-dimensional hydrogen.
Using $a_0=\hbar^2/m\et$, we finally obtain
\begin{equation}\label{inf-estim}
\frac{\hbar^2}{2m} \min \sigma( \Hc)
\leq 
\frac{\hbar^2}{2m}h[\psi]=-\frac{\hbar^2}{2m}\frac1{a_0^2}=-\frac{me_3^4}{2\hbar^2}.
\end{equation}
\end{proof}

\appendix
\section{}

\counterwithin{theorem}{section}
 \setcounter{theorem}{0}

\begin{theorem}[The first representation theorem, {\cite[Thm.VI.2.1]{Kato-1966}}]\label{thm.kato}
	Let $h: \Dom(h) \times \Dom(h) \to \C$ be a densely defined, symmetric, bounded from below and closed sesquilinear form in $\H$. Then there exists a self-adjoint operator $H$ such that 
	\begin{enumerate}[i)]
		 \setlength{\itemsep}{0pt}
		 \setlength{\parskip}{0pt}
	\item $\Dom(H) \subset \Dom(h)$ and $h(\phi,\psi) = \langle \phi, H \psi \rangle$ for every $\phi \in \Dom(h)$ and $\psi \in \Dom(H)$;
	\item $\Dom(H)$ is a core of $h$;
	\item if $\psi \in \Dom(h)$, $\eta \in \H$, and $h(\phi, \psi) = \langle \phi, \eta \rangle$ holds for every $\phi$ belonging to a core of $h$, then $\psi \in \Dom(H)$ and $H \psi = \eta$. The self-adjoint operator $H$ is uniquely determined by the \mbox{condition $i)$}.
	\end{enumerate}
\end{theorem}

\begin{theorem}[Weyl's criterion, {\cite[Thm.~VII.12]{Reed1}}]\label{thm.weyl}
	Let $H$ be a self-adjoint operator on $\H$. A point $\lambda$ belongs to $\sigma(H)$ if, and only if, there exists a sequence $\{\psi_n\}_{n\in \N} \subset \Dom(H)$ such that $\|\psi_n\|=1$ for all $n\in \N$ and $\lim_{n\to\infty} \|(H-\lambda)\psi_n\|\to 0$.
	Moreover, $\lambda$ belongs to $\sigma_{\rm ess}(H)$ if, and only if, in addition to the above properties the $\{\psi_n\}$ converges weakly to zero in $\H$.
\end{theorem}

\begin{theorem}[KLMN, {\cite[Thm.~X.17]{Reed2}}]\label{thm.KLMN}
Let $h_0: \Dom(h_0) \times \Dom(h_0) \to \C$ be a densely defined, symmetric, non-negative and closed sesquilinear form in $\mathcal{H}$. Let $v$ be a symmetric sesquilinear form satisfying
\begin{equation}
\begin{aligned}
1. & \quad	 \Dom(h_0) \subset \Dom(v),  \\
2. & \quad	 \forall \psi \in \Dom(h_0), \quad |v[\psi]| \leq a \, h_0[\psi] + b \, \|\psi\|^2,		 
\end{aligned}
\end{equation}
where $a$, $b$ are non-negative and $a<1$. 
Then there exists a unique self-adjoint and bounded from below operator $H$, associated with the closed symmetric sesquilinear form 
\begin{equation}
h:=h_0+v, \qquad \Dom(h):=\Dom(h_0).
\end{equation}
\end{theorem}

{\footnotesize
\bibliographystyle{ieeetr}
}
\bibliography{bibliography-r3xs1}

\end{document}